%% file: main.tex
\documentclass{article}
\usepackage{graphicx} 

\title{A Short Note on Algorithms For The Noisy $k$-XOR Problem}
\author{Kiril Bangachev\\
\small{MIT Department of EECS}\footnote{Funded by an NAE Grand Challenge Vest Fellowship and a Schwarzman College of Computing Fellowship for the duration of this work.}}
\date{October 2026}
\input{macros}

\usepackage{float}
\usepackage{subcaption}
\usetikzlibrary{arrows.meta}
\definecolor{algI}{HTML}{8753A1}
\definecolor{algII}{HTML}{B97714}
\definecolor{algIII}{HTML}{3266A8}
\definecolor{algIV}{HTML}{C8443D}
\definecolor{moore}{HTML}{2A8653}

\bibliography{ref}
\begin{document}
\maketitle
\begin{abstract}
    This mostly expository note studies the sample-noise regimes in which the testing problem for noisy $k$-XOR in dimension $n$ can be efficiently solved. Concretely, we ask what pairs $(\beta, \delta)$ admit algorithms running in time $n^{O(k)}$ where $m = n^\beta, \beta \in [0,k]$ is the number of samples and $\mathsf{Bern}(n^{-\delta}),\delta \in (0,+\infty)$ is the noise distribution. The current state-of-the-art algorithms are: I. A folklore algorithm when $\beta > 1, \delta>1$ based on Gaussian elimination; II. A folklore algorithm when $\beta\ge k/2$ based on the birthday paradox; III. A recent algorithm in the regime $\beta > \delta + k(1-\delta), \delta \in (0,1)$ based on a reduction to the first algorithm \cite{dao24sparsedense,chen2024algorithmssparselpnlspn}. IV. Extending \cite{alonfeigeprobes}, we strictly improve the last regime to $\beta > \delta + (k-1)(1-\delta), \delta \in (0,1).$ These algorithms correspond to efficiently finding short cycles in uniform hypergraphs and can be equivalently phrased as algorithmic Moore bounds.
\end{abstract}

\begin{figure}[H]
  \centering
\scalebox{0.7}{%
\begin{tikzpicture}[x=2.1cm,y=4.3cm,font=\normalsize,
I/.style={algI,line width=1.25pt,dash pattern=on 5pt off 2.8pt},
II/.style={algII,line width=1.25pt,dash pattern=on 2pt off 2.2pt},
III/.style={algIII,line width=1.5pt},
IV/.style={algIV,line width=1.5pt},
M/.style={moore,line width=1.4pt,dash pattern=on 6pt off 2.5pt on 1.2pt off 2.5pt}]
\node[font=\normalsize\bfseries] at (3,1.53) {Noise versus samples};
\fill[black!9] (0,0) rectangle (1,1.35);
\node[rotate=90,font=\normalsize,text=black!65] at (.48,.67) {Information-theoretically impossible};
\fill[moore!9] (1,0)--(1,1)--(3,0)--cycle;
\fill[purple!12] (1,1)--(3,.5)--(3,0)--cycle;
\draw[black!15,line width=.4pt] (0,1)--(6.12,1);
\draw[I] (1,1.35)--(1,1)--(6.1,1);
\draw[II] (3,0)--(3,1.35);
\draw[III] (1,1)--(6,0);
\draw[IV] (1,1)--(5,0);
\draw[M] (1,1)--(3,0);
\fill[black!75] (1,1) circle[radius=1.65pt];
\node[text=moore!90!black,align=center,font=\normalsize] at (1.62,.24) {Moore Bound\\ (unknown for $k$-XOR)};
\node[font=\normalsize] at (2.43,.49) {unknown};
\draw[-{Stealth[length=5pt]},line width=.7pt] (0,0)--(6.3,0);
\draw[-{Stealth[length=5pt]},line width=.7pt] (0,0)--(0,1.39);
\foreach \x/\lab in {0/0,1/1,3/{k/2},5/{k-1},6/k} {
 \draw (\x,0)--++(0,-.014);
 \node[anchor=north,yshift=-3pt] at (\x,0) {$\lab$};
}
\foreach \y/\lab in {0/0,1/1} {
 \draw (0,\y)--++(-.045,0);
 \node[anchor=east,xshift=-3pt] at (0,\y) {$\lab$};
}

\node at (3.1,-.17) {Sample exponent $\beta$ \quad ($m=n^\beta$)};
\node[rotate=90,align=center] at (-.76,.70)
  {Noise exponent $\delta$ ($\eta=n^{-\delta}$)\\
   (or, finding cycles of length $\widetilde O(n^\delta)$)};
\draw[I] (0,-.37)--(.31,-.37);
\node[anchor=west,align=left,font=\normalsize] at (.39,-.37) {\textbf{I. Gaussian elimination}\\$\beta\geq 1,\quad \delta\geq 1$};
\draw[II] (3.27,-.37)--(3.58,-.37);
\node[anchor=west,align=left,font=\normalsize] at (3.66,-.37) {\textbf{II. Birthday paradox}\\$\beta=k/2$};
\draw[III] (0,-.61)--(.31,-.61);
\node[anchor=west,align=left,font=\normalsize] at (.39,-.61) {\textbf{III. \cite{dao24sparsedense,chen2024algorithmssparselpnlspn}}\\$\beta=\delta+k(1-\delta)$};
\draw[IV] (3.27,-.61)--(3.58,-.61);
\node[anchor=west,align=left,font=\normalsize] at (3.66,-.61) {\textbf{IV. This note (extending \cite{alonfeigeprobes})}\\$\beta=\delta+(k-1)(1-\delta)$};
\draw[M] (0,-.85)--(.31,-.85);
\node[anchor=west,align=left,font=\normalsize] at (.39,-.85) {\textbf{Moore Bound}\\$\beta=\delta+\tfrac{k}{2}(1-\delta)$};
\node[anchor=west,align=left,font=\normalsize,text=black!60] at (3.66,-.85) {Bounds shown at exponent level.\\Schematic spacing: $k=6$.};
\end{tikzpicture}%
}
  \caption{Sample--noise tradeoffs for noisy $k$-XOR, equivalently interpreted
    through finding cycles of length $\widetilde O(n^\delta)$. Regions strictly above and to the right of the curves are ``easy'' and handled by the respective algorithms. The Moore bound does not imply hardness for noisy $k$-XOR but is an obstruction to the meta-approach described in Lemma~\ref{lemma:attack}, which underlies all four algorithms.}
  \label{fig:kxor-tradeoffs}
\end{figure}
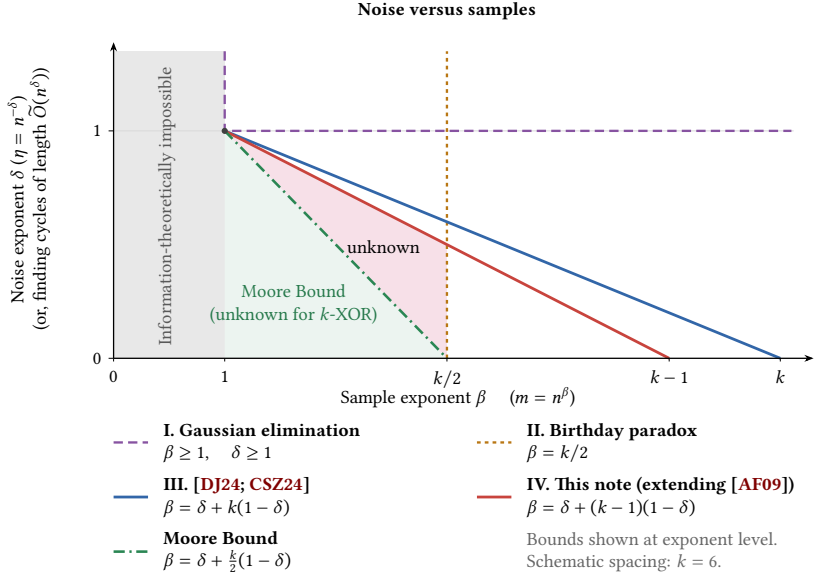
\section{Introduction}

The noisy $k$-XOR problem is informally ``linear regression over $\field$ where each sample is uniformly random $k$-sparse'' (we define it rigorously below). Perhaps surprisingly, unlike standard linear regression over $\mathbb{R},$ the problem exhibits a substantial information-computation gap when $k \ge 3.$ In some ``information-theoretically possible regimes'', there is evidence for near-exponential hardness of the problem \cite{schoenebeck2008linear}. The simple linear structure combined with strong hardness makes the noisy $k$-XOR problem useful across
Learning Theory \cite{raghavendra17stronglyrefuting,feige02refutingkxor,feige06witnesses,Allen2015HowTR,guruswami23semirandom, bogdanov19XORcodes,bogdanov25searchtodecision,blasiok2026expanders}, Cryptography \cite{applebaum2010public,dao24sparsedense,jain2022indistinguishability,ragavan24io}, and Average-Case Complexity \cite{grigoriev2001linear,schoenebeck2008linear}.

Despite the prominence of the problem,  we still don't know the answer to a basic question: \emph{In what regimes is noisy $k$-XOR easy?} The goal of this note is to summarize the (surprisingly limited) progress towards this question. We concretely focus on the regime of small sparsity $k = \tilde{O}(1),$ inverse polynomial noise rate $\eta = n^{-\delta}, \delta\in (0,+\infty)$ and polynomially many equations $m = n^\beta, \beta\in (0,k].$ We also interpret \emph{easy} as the existence of algorithms running in time $n^{O(k)}.$ This is for two reasons. First, a lot of the interesting behavior is when $\beta = \Theta(k)$ (Figure~\ref{fig:kxor-tradeoffs}), so time $n^{O(k)}$ is polynomial in the input size. Second, the recent work \cite{bangachev2024nearoptimaltimesparsitytradeoffssolving} gives evidence for an $n^{\Omega(k)}$ running-time lower bound on the noisy $k$-XOR problem.\footnote{The result of \cite{bangachev2024nearoptimaltimesparsitytradeoffssolving} establishes $n^{\Omega(\frac{k}{\mathsf{poly}(\log \log n, \log k)})}$ hardness and only applies to high-noise regimes.} The problem has testing, search, and refutation versions. We focus on testing and discuss the other two at the end.

\begin{problem}[Testing Problem For Noisy $k$-XOR]
\label{problem:mainproblem}
Suppose that $m,n,k$ are natural numbers with $1\le k \le n,$ $m = n^\beta$ for some $\beta \in (0,+\infty).$ 
Let 
$\sparserow(\field)$ be the uniform distribution over $k$-sparse vectors in $\field^n$.
Suppose also that $\delta\in (0, +\infty)$ and the noise is $\eta = n^{-\delta}.$ The testing problem is between the following two distributions over $(A,y)\in \field^{m \times n} \times \field^{m}$
:
\begin{equation}
    \begin{split}
        & \mathsf{H}_0 \; : \; A\sim \sparserow(\field)^{\otimes m}\perp y \sim \unif(\field^{m}),\\
        & \mathsf{H}_1 \; : \; A\sim \sparserow(\field)^{\otimes m}\perp s \sim \unif(\field^{n})\perp  e\sim \mathsf{Bern}(\eta)^{\otimes m}= \mathsf{Bern}(n^{-\delta})^{\otimes m}\;,\text{ and } y = As + e.
    \end{split}
\end{equation}
The goal is to design an algorithm $\mathcal{A}: \field^{m \times n} \times \field^{m}\longrightarrow\{0,1\}$ running in time $n^{O(k)}$ and distinguishing the two distributions with constant advantage, i.e.
$$
\Big| \prob_{(A,y)\sim \mathsf{H}_0}[\mathcal{A}(A,y) = 0] - \prob_{(A,y)\sim \mathsf{H}_1}[\mathcal{A}(A,y) =0]\Big|\ge 1/100.
$$
\end{problem}
The constant $1/100$ is unimportant and can be boosted to any other constant by taking constantly many more samples, or to $1 - 1/\mathsf{poly}(n)$ by taking $O(\log n)$ times more samples. Since we care only about the exponent $\beta,$ this change of the advantage is inconsequential.

We mention, without proof, that whenever $m \ge 2n,$ there is a consistent test between $\mathsf{H}_0$ and $\mathsf{H}_1$ based on the largest number of simultaneously satisfiable equations. This note is dedicated to exploring when a ``computationally efficient'' test with this property exists.

\paragraph{The Attack: Even Cycles and An Algorithmic Moore Bound.} A natural approach to Problem~\ref{problem:mainproblem} is based on finding linearly dependent rows of $A$ (also known as even covers or cycles). This meta-approach underlies many attacks on noisy $k$-XOR and related problems including the Kikuchi matrix approach \cite{Allen2015HowTR,ghd25classicallspeeedup,agrawal2026attacks} and the BKW algorithm \cite{BKW,lyubashesky05lpn,kothari2026strongly} in addition to the simple algorithms discussed below.

\begin{lemma}
\label{lemma:attack}
Suppose that there is an algorithm running in
time $T$ which on input $A\in \field^{ m\times n}$ 
finds with probability at least $1/10$ a set of
$1\le L \le n^\delta$ rows of $A,$ indexed by $I \subseteq [m]$ that sum up to 
$\mathbf{0}_n \in \field^n.$ We claim that for any constant $\epsilon>0, $ there is an $O(T + \mathsf{poly}(n))$ testing algorithm for the noisy $k$-XOR problem with $m = n^\beta, \eta = n^{-\delta-\epsilon}.$\end{lemma}
\begin{proof}
Let $I$ be the set of coordinates returned by the algorithm. Then, under $\mathsf{H}_0,$ we have that $1_I^Ty\sim \mathsf{Bern}(1/2)$ as $y$ is independent of $A$ and has unbiased coordinates. Under $\mathsf{H}_1,$ on the other hand, 
$1_I^Ty = 1_I^T(As + e) = 1_I^Te.$ This is the XOR of at most $ n^\delta$ independent $\mathsf{Bern}(n^{-\delta-\epsilon})$ variables, which is $\mathsf{Bern}(O(n^{-\epsilon})).$ This distribution can be easily distinguished from $\mathsf{Bern}(1/2)$ with advantage $1/10.$
\end{proof}
All the algorithms below and in Figure~\ref{fig:kxor-tradeoffs} are based on Lemma~\ref{lemma:attack}.

We note that the shortest cycle given $m = n^\beta$ equations is of size roughly $\max (2, \tilde{O}(n^{\frac{k-2\beta}{k-2}})).$ Thus, a shortest-cycle oracle would succeed in the regime $\beta\ge \delta + k/2(1 - \delta).$ This bound is known as the Moore bound and is known to hold also for worst-case $A$ (related to the semi-random version of the problem, see Section~\ref{sec:semirandom}) \cite{hsieh2023simple,bandeira2026hypergraphmoorebound,feige2008small}. When such cycles can be found in time $n^{O(k)}$ is a wide-open problem for both random and worst-case $A$. Even more mysterious is the regime when the noise is $n^{-\delta}$ but cycles of length $\tilde{O}(n^\delta)$ do not exist, so one cannot rely on the approach of Lemma~\ref{lemma:attack} even with a cycle-finding oracle.\footnote{The  BKW algorithm \cite{BKW} for Learning Parity with Noise operates in a regime where cycles of length $1/\eta$ do not exist. 
Instead, it finds (nearly exponentially) many longer cycles. Such an approach seems more difficult towards our problem since we only allow time $n^{O(k)}.$ The number of cycles needed scales exponentially with cycle length beyond length $O(1/\eta),$ see \cite{BKW}.}

We outline as an interesting problem whether the converse to Lemma~\ref{lemma:attack} holds. 

\begin{openproblem}
    Suppose that there exists an algorithm which solves the testing Problem~\ref{problem:mainproblem} for some $(\beta, \delta)$ in time $T = T_{n,k,\beta, \delta}.$ Does there exist an algorithm running in time $\mathsf{poly}(T, n^k)$ which on input $A\sim \sparserow(\field)^{\otimes m}$ with $m = n^\beta$ finds a non-trivial cycle of length $n^{\delta + o(1)}$?
\end{openproblem}

\section{Algorithms}
We now present the four algorithms from Figure~\ref{fig:kxor-tradeoffs} which are all based on Lemma~\ref{lemma:attack}.

\paragraph{I. Low-Noise Algorithm [Folklore].} When $m \ge n + 1,$ the first $n+1$ rows of $A$ are linearly dependent. Hence, there is a cycle of length $\le n+1.$ It can be efficiently discovered in time $n^{O(1)}$ via Gaussian elimination. Combined with Lemma~\ref{lemma:attack}, this shows that the regime $\beta \ge 1, \delta \ge 1$ is easy.

\paragraph{II. Birthday Paradox Algorithm [Folklore].} Another folklore algorithm is based on the birthday paradox. Suppose that $\beta \ge k/2.$ As the number of distinct rows is $n^k$, with a constant positive probability, there exist two rows among the first $n^{k/2}$ rows of $A$ that coincide. Those two rows form a cycle. Combined with Lemma~\ref{lemma:attack}, this shows that the regime $\beta \ge k/2, \delta > 0$ is easy. 

\paragraph{III. Algorithm of \cite{dao24sparsedense,chen2024algorithmssparselpnlspn}.} In these two recent works, the authors adapt the folklore algorithm resolving $\beta \ge 1, \delta \ge 1$ to the regime $\beta > k - (k-1)\delta, \delta \in (0,1).$ The algorithm is the following. Consider only the first $n^\delta$ columns of $A.$ Then, maintain only the rows of $A$ whose support is in those $n^\delta$ columns. A simple first and second moment calculation shows that the expected number of samples retained is $\ge n^\delta + 1.$ Restricting to those $n^\delta + 1$ samples and $n^\delta$ columns we can apply the first algorithm. 

For completeness, we give the first moment calculation. Independence across samples then implies the result. The probability that the support of a sample falls fully in the first $n^\delta$ coordinates is $\frac{\binom{n^\delta}{k}}{\binom{n}{k}} = \Theta\big(\frac{n^\delta}{n}\big)^k.$ Taking the expectation over the $m$ samples, we get an expected number of samples with support in the first $n^\delta$ columns 
\begin{equation}
    \begin{split}
        m \times \Theta\big(\frac{n^\delta}{n}\big)^k = 
        n^{\beta + k(\delta-1) +o(1)}.
    \end{split}
\end{equation}
Hence, whenever $\beta >\delta + k(1-\delta),$ there are at least $n^\delta +1$ samples with support in the first $n^\delta$ columns as desired. 

\paragraph{IV. Improved Algorithm Based on \cite{alonfeigeprobes}.} Finally, we present our algorithm which improves the regime of \cite{dao24sparsedense,chen2024algorithmssparselpnlspn} to $\beta > \delta + (k-1)(1-\delta), \delta \in (0,1).$ It is a relatively straightforward extension of an algorithm for $k=3$ in \cite[Lemma 3.3]{alonfeigeprobes}.

We still separate the first $n^\delta$ columns as in algorithm \textbf{III}. However, this time we keep samples with exactly two coordinates outside of those $n^\delta$ columns (as opposed to 0). This is useful since those samples with exactly 2 coordinates on the remaining $n-n^\delta$ columns can be viewed as edges in a multigraph indexed by the $n-n^\delta$ columns. In graphs, unlike in hypergraphs, cycles are easy to discover. 

A similar first moment + independence argument shows that whenever $\beta > \delta + (k-1)(1-\delta)+\epsilon,$ with high probability there are at least $n^{1 + \epsilon}$ samples which have exactly $2$ coordinates outside of the first $n^\delta$ columns. 

We now view those samples as a graph on $n - n^\delta$ vertices and $n^{1 + \epsilon}$ edges. The following fact is well known and we only sketch its proof for completeness.

\begin{fact}
Suppose that $G = (V,E)$ is a multigraph on $p$ vertices and $p^{1 + \epsilon}$ edges for some absolute constant $\epsilon>0.$ Then, for all large enough $p,$ there exists a collection of $p$ edge-disjoint cycles each of length $O(1/\epsilon).$
\end{fact}
\begin{proof}[Proof Sketch] It is enough to show the existence of one such cycle. Then, we can remove its edges and induct. To show the existence of one such cycle, notice first that $G$ has a subgraph $G' = (V', E')$ with minimal degree at least $p^\epsilon/2.$ To find one such $V',$ keep deleting vertices of degree less than $p^\epsilon/2.$ Note that in the process at most $p \times p^\epsilon/2 < |E|$ edges are deleted. Hence, there are some edges which are not deleted and the process terminates with a non-empty graph. 

This graph has at most $p$ vertices and minimal degree $p^\epsilon/2.$ Hence, a simple breadth-first search procedure identifies a cycle of length at most $O(\log_{p^\epsilon/2}p) = O(1/\epsilon).$
\end{proof}

The proof above not only shows that there are $n$
edge-disjoint cycles on the last $n - n^\delta$ columns, but it also provides an $n^{O(1)}$ algorithm for finding them. 

Let the cycles correspond to equations indexed by disjoint sets respectively $I_1, I_2, \ldots, I_n,$ each of length $O(1/\epsilon).$ Then, the row-vectors $1_{I_1}^TA,1_{I_2}^TA, \ldots 1_{I_n}^TA $ fully belong to the first $n^\delta$ 
columns. Again, Gaussian elimination will give a non-trivial cycles among at most $n^\delta$ of them. This leads to a non-trivial cycle of length $O(n^\delta/\epsilon)$ 
of the original equations which is enough to apply Lemma~\ref{lemma:attack}.

\section{Some Remarks}
\paragraph{Search and Strong Refutation.}
This note focused on the testing problem. However, search and refutation are also of interest to the high-dimensional statistics,  
learning theory, and 
cryptography communities. 

We note that up to a blow-up of the sample size by a multiplicative factor of $n$ (or less) each of these algorithms can be easily transformed into a search algorithm (that is, finding $s$ under the promise that the distribution is $\mathsf{H}_1.$ For even $k$ this is only possible up to a global flip of the coordinates of $s$). This can be either seen directly for the simple algorithms we presented or via the recent search (in noisy $k$-XOR)-to-decision (in noisy $(k-1)$-XOR) reduction \cite{bogdanov25searchtodecision}. Prior reductions with worse sample blow-up are also known \cite{bogdanov19XORcodes}.

While there doesn't seem to be a qualitative difference between search and testing, strong refutation is significantly more challenging (see \cite{raghavendra17stronglyrefuting} for the definition of refutation). In fact, Sum-of-Squares (SoS) lower bounds suggest that when $\beta\le k/2,$ there are no $n^{O(k)}$ refutation algorithms. 

It is worthwhile noting that the recent refutation algorithm \cite{kothari2026strongly} for learning parity with noise, extending the BKW algorithm, surpasses known SoS lower bounds. 
\label{sec:semirandom}

\paragraph{Other Parameter Regimes.} Our focus is on the setting of $n^{O(k)}$ time budget, $\beta \le k/2, \eta = n^{-\Omega(1)}.$ However, of interest are many other regimes, such as when quasi-polynomial or even subexponential time algorithms are permissible  
\cite{raghavendra17stronglyrefuting,blasiok2026expanders,mao2026near}, or the noise is of the form $\mathsf{Bern}(1/2 - n^{-\kappa})$ and the number of equations is $m = n^\beta$ for $\beta \in [k/2, k)$ 
\cite{bresler2026average} (see Section 2.1 of \cite{bresler2026average} for a similar discussion of hardness and polynomial time algorithms in that regime). 

\paragraph{Semirandom Instances.} The current note focuses on the case when $A\sim \sparserow(\field)^{\otimes m}.$ Common in the literature is also the semi-random case in which $A$ is worst-case, but independent of the secret and noise \cite{guruswami23semirandom,abascal2021strongly}. We note that most of the discussion in the current note extends rather directly to this case as well. Neither of the algorithms relies strongly on the fact that $A$ is random. Namely, algorithm I applies directly to any $A$ (does not even need to be sparse). Algorithm II still applies as in any $k$-uniform hypergraph with $n^{k/2+\epsilon}$ edges, there is a cycle of length $O(1/\epsilon)$ which can be found in time $n^{O(k)}$ \cite{feige2008small}. Algorithms III and IV need to be modified as follows. Note that if instead of the first $n^\delta$ columns one takes a uniformly random set of $n^\delta$ columns, the expected number of samples with 0 (respectively 2) coordinates outside of those columns is the same $mp$ with $p \approx n^{-(1-\delta)k}$ (respectively $p \approx n^{-(1-\delta)(k-2)}$). However, as samples are not independent, we cannot rely on a second moment concentration argument. Instead we notice that the maximal number of samples with the desired property is $m.$ Hence, the probability that there are at least  $mp/2$ samples with 0 (respectively 2) coordinates outside of those columns is at least  $\frac{mp/2}{m-mp/2} = \Theta(p) = n^{-\Theta(k)}.$ After $n/p$ trials we will find with high probability $mp/2$ rows that satisfy the desired property.

\paragraph{Hardness in Restricted Models of Computation.} One reason why the complexity of the problem is still poorly understood is that we do not have any reasonable hardness predictions for the testing/search problem against restricted classes of algorithms such as via the Low-Degree Polynomials framework, Sum-of-Squares Hierarchy, or the Overlap Gap Property. While one can prove lower bounds in any of these frameworks, those lower bounds are meaningless as even the simplest algorithms -- such as III and IV -- do not fall in the restricted classes of algorithms (they are highly non-local and high-degree due to the Gaussian elimination step). Hence, it is not surprising that the predictions are wrong. For instance, low-degree and Sum-of-Squares predict hardness for any noise level when $\beta <k/2.$ This is an expected failure mode of the two paradigms in algebraic settings (formally, the Sum-of-Squares prediction is not known to be wrong as it predicts hardness of refutation). The interesting question that it raises is:

\begin{openproblem}
    Design a noise-sensitive hardness heuristic that will predict in what regimes Problem~\ref{problem:mainproblem} is hard without contradicting any of the algorithms I, II, III, IV. Or, show that all of the ``unknown'' regime in Figure~\ref{fig:kxor-tradeoffs} is easy in time $n^{O(k)}.$
\end{openproblem}

Assuming that the problem is not easy at least in part of the ``unknown regime'', we strongly believe that such a heuristic will be broadly useful towards other ``algebraic'' problems in high-dimensional statistics. 

Alternatively, it would be interesting to prove hardness via reductions from well-established assumptions.

\paragraph{Acknowledgments.} I am grateful to Guy Bresler, Stefan Tiegel, and Vinod Vaikuntanathan for many insightful conversations on noisy $k$-XOR and related problems.

\paragraph{AI Usage Statement.} AI was used to discover the connection to \cite{alonfeigeprobes}. AI was also used to revise the document.

\printbibliography

\end{document}

%% file: macros.tex
\usepackage[utf8]{inputenc}
\usepackage{amsmath}
\usepackage{amssymb,amsfonts,amsthm}
\usepackage{mathrsfs}
\usepackage{mathtools}
\usepackage{algorithm}
\usepackage{algorithmic}
\usepackage{xcolor}
\usepackage{tikz}
\usepackage{fullpage}
\usepackage{dsfont}
\usepackage{thmtools}
\usepackage{thm-restate}
\usepackage{enumitem}
\usepackage{graphicx}
\usepackage{sidecap}
\usepackage{bbm}
\usepackage{array}
\newcolumntype{M}[1]{>{\centering\arraybackslash}m{#1}}

\usepackage{accents}
\usepackage{lmodern}

\usepackage{microtype}      
\usepackage{libertine}
\usepackage{libertinust1math}

\usepackage{hyperref}
\usepackage{xcolor}
\definecolor{maroon}{rgb}{0.5, 0.0, 0.0}
\hypersetup{
  colorlinks   = true, 
  urlcolor     = blue, 
  linkcolor    = blue, 
  citecolor   = maroon 
}
\usepackage[capitalize, nameinlink]{cleveref}

\usepackage{tikz}
\usetikzlibrary{decorations.pathreplacing}
\usetikzlibrary{arrows,positioning}

\usepackage{comment}
\DeclareRobustCommand\shape{
 \lower5pt\hbox{
 \hskip-7pt
  \tikzset{circ/.style={circle, draw, fill=black, scale=.15}}
  \begin{tikzpicture}[semithick,scale=.3]
  \node (l1) at (0,.5) [circ]{};
  \node (l3) at (0.5,0.3) [circ]{};
  \draw[-] (l1) to node [auto] {} (l3);
    \end{tikzpicture}
  \hskip-8pt}
}

\makeatletter
\newcommand{\mylabel}[2]{#2\def\@currentlabel{#2}\label{#1}}
\makeatother

\newcommand{\unif}{\mathsf{Unif}}

\newtheorem*{theorem*}{Theorem}

\newtheorem{lemma}{Lemma}[section]
\newtheorem*{lemma*}{Lemma}

\newtheorem{problem}{Problem}[]
\newtheorem{openproblem}{Open Problem}[]

\newtheorem{fact}{Fact}[]

\DeclareMathOperator*{\E}{{\rm I}\kern-0.18em{\rm E}}

\newcommand{\prob}{\Pr}

\renewcommand{\Pr}{\,{\rm I}\kern-0.18em{\rm P}}

\definecolor{todored}{HTML}{AD0029}

\newcommand\MYcurrentlabel{xxx}
\newcommand{\MYstore}[2]{%
  \global\expandafter \def \csname MYMEMORY #1 \endcsname{#2}%
}
\newcommand{\MYload}[1]{%
  \csname MYMEMORY #1 \endcsname%
}
\newcommand{\MYnewlabel}[1]{%
  \renewcommand\MYcurrentlabel{#1}%
  \MYoldlabel{#1}%
}
\newcommand{\MYdummylabel}[1]{}
\newcommand{\torestate}[1]{%
  \let\MYoldlabel\label%
  \let\label\MYnewlabel%
  #1%
  \MYstore{\MYcurrentlabel}{#1}%
  \let\label\MYoldlabel%
}
\newcommand{\restatetheorem}[1]{%
  \let\MYoldlabel\label
  \let\label\MYdummylabel
  \begin{theorem*}[Restatement of \cref{#1}]
    \MYload{#1}
  \end{theorem*}
  \let\label\MYoldlabel
}
\newcommand{\restatelemma}[1]{%
  \let\MYoldlabel\label
  \let\label\MYdummylabel
  \begin{lemma*}[Restatement of \cref{#1}]
    \MYload{#1}
  \end{lemma*}
  \let\label\MYoldlabel
}
\newcommand{\restateprop}[1]{%
  \let\MYoldlabel\label
  \let\label\MYdummylabel
  \begin{proposition*}[Restatement of \cref{#1}]
    \MYload{#1}
  \end{proposition*}
  \let\label\MYoldlabel
}
\newcommand{\restatefact}[1]{%
  \let\MYoldlabel\label
  \let\label\MYdummylabel
  \begin{fact*}[Restatement of \cref{#1}]
    \MYload{#1}
  \end{fact*}
  \let\label\MYoldlabel
}
\newcommand{\restateparameters}[1]{%
  \let\MYoldlabel\label
  \let\label\MYdummylabel
  \begin{parameters*}[Restatement of \cref{#1}]
    \MYload{#1}
  \end{parameters*}
  \let\label\MYoldlabel
}
\newcommand{\restate}[1]{%
  \let\MYoldlabel\label
  \let\label\MYdummylabel
  \MYload{#1}
  \let\label\MYoldlabel
}

\usepackage[backend=biber,
backref=true,
style=alphabetic,
minalphanames=3,
maxalphanames=3,
maxnames=50
]{biblatex}
\newcommand{\sparserow}{\mathsf{Sp}_{n,k}}
\newcommand{\field}{\mathbb{F}_2}